\documentclass[11pt,reqno]{amsart}
\usepackage{amssymb}
\usepackage{graphics}
\usepackage{xcolor}
\usepackage{bbm}
\usepackage{enumerate}
\usepackage[pagebackref, colorlinks = true, linkcolor = blue, urlcolor  = blue, citecolor = red]{hyperref}
\usepackage[margin=1in]{geometry}
\usepackage{enumitem}

\newtheorem{theorem}{Theorem}[section]
\newtheorem{definition}[theorem]{Definition}
\newtheorem{proposition}[theorem]{Proposition}
\newtheorem{corollary}[theorem]{Corollary}
\newtheorem{lemma}[theorem]{Lemma}
\newtheorem{remark}[theorem]{Remark}

\newcommand{\id}{\ensuremath{\mathbbm{1}}}

\def\beq{\begin{equation}}
\def\eeq{\end{equation}}
\def\bq{\begin{quote}}
\def\eq{\end{quote}}
\def\ben{\begin{enumerate}}
\def\een{\end{enumerate}}
\def\bit{\begin{itemize}}
\def\eit{\end{itemize}}

\def\ra{\rightarrow}

\def\lb{\left(}
\def\rb{\right)}
\def\lset{\lbrace}
\def\rset{\rbrace}

\def\l|{\left|}
\def\r|{\right|}
\def\lbr{\left[}
\def\rbr{\right]}

\def\one{\id}

\newcommand\C{\mathbb{C}}

\newcommand\R{\mathbb{R}}
\newcommand\N{\mathbb{N}}

\begin{document}

\title[Restr. on the Schmidt rank of bipart. unitary op. beyond dim. two]{Restrictions on the Schmidt rank of bipartite unitary operators beyond dimension two}

\author{Alexander M\"uller-Hermes}
\address{AMH: Department of Mathematical Sciences, University of Copenhagen, Universitetsparken 5,  2100 Copenhagen, Denmark}
\email{muellerh@math.ku.dk and muellerh@posteo.net}

\author{Ion Nechita}
\address{IN: Zentrum Mathematik, M5, Technische Universit\"at M\"unchen, Boltzmannstrasse 3, 85748 Garching, Germany
and CNRS, Laboratoire de Physique Th\'{e}orique, IRSAMC, Universit\'{e} de Toulouse, UPS, F-31062 Toulouse, France}
\email{nechita@irsamc.ups-tlse.fr}

\subjclass[2000]{}
\keywords{}

\begin{abstract}
There are none.
\end{abstract}

\date{\today}

\maketitle

\tableofcontents

\section{Introduction}

In this paper, we study the possible values of an integer invariant of bipartite operators, the \emph{operator Schmidt rank}. 
For a non-zero operator $X\in \mathcal{M}_n(\mathbb C) \otimes \mathcal{M}_m(\mathbb C)$ the operator Schmidt rank is defined as the unique number $\Omega\lb X\rb\in\N$ such that
\[
X = \sum^{\Omega(X)}_{i=1} A_i\otimes B_i
\]
with orthogonal\footnote{with respect to the Hilbert-Schmidt inner product $\left\langle X,Y\right\rangle = \text{tr}\left( X^* Y\right)$ for $X,Y\in\mathbb{M}_n\left(\mathbb{C}\right)$} sets of non-zero operators $\left\lbrace A_i\right\rbrace^r_{i=1}\subset \mathcal{M}_n\left(\mathbb{C}\right)$ and $\left\lbrace B_i\right\rbrace^r_{i=1}\subset \mathcal{M}_m\left(\mathbb{C}\right)$. Note that $\Omega(X)\in\left\lbrace 1,\ldots ,\min(n,m)^2\right\rbrace$, and for general operators $X$ all these ranks can occurr. Here we are interested in the case of unitary operators $X$. The requirement of $X$ being unitary causes an interesting restriction in dimension two:

\begin{theorem}[\cite{dvc}]
For $U\in \mathcal{U}\left(\mathbb{C}^2\otimes \mathbb{C}^2\right)$ only ranks $\Omega(U)\in\left\lbrace 1,2,4\right\rbrace$ are possible. 
\end{theorem}

The proof of the result above relies on the so-called \emph{normal form} of two-qubit unitary operations from \cite{kci}. In \cite{nie} further restrictions on the operator Schmidt ranks of higher dimensional unitary operators have been conjectured. However, in dimensions $n=m=3$ no such restrictions have been observed:

\begin{theorem}[{\cite[Theorem 10]{tys}}]
For $U\in \mathcal{U}\left(\mathbb{C}^3\otimes \mathbb{C}^3\right)$ all ranks $\Omega(U)\in\left\lbrace 1,2,\ldots , 9\right\rbrace$ are possible. 
\end{theorem}

These results raise the question, whether there are any restrictions on operator Schmidt ranks for bipartite unitaries in higher dimensions. In this paper, we prove the following result:

\begin{theorem}\label{thm:main}
For $U\in \mathcal{U}\left(\mathbb{C}^n\otimes \mathbb{C}^m\right)$ and $(n,m)\neq (2,2)$ all ranks $\Omega(U)\in\left\lbrace 1,2,\ldots , \min(n,m)^2\right\rbrace$ are possible.
\end{theorem}
\begin{proof}
The proof consists of a collection of examples, detailed in the following sections. First, we show that the result holds in the ``balanced'' case $n=m$, as follows:
\begin{itemize}
\item Ranks $\{1, \ldots, n\}$: Proposition \ref{prop:diag-blocks}
\item Ranks $\{n, \ldots, n^2\} \setminus \{n+1,n^2-1\}$: Theorem \ref{thm:permutation}
\item Rank $n+1$, for odd $n$: Corollary \ref{cor:product-odd}
\item Rank $n+1$, for even $n$: Proposition \ref{prop:product-even}
\item Rank $n^2-1$: Proposition \ref{prop:existence-fourier}.
\end{itemize}
The general case $n \neq m$ is proven in Theorem \ref{thm:unbalanced}.
\end{proof}

\bigskip

\noindent \emph{Acknowledgments.} We thank Chris Perry for interesting discussions. AMH acknowledges financial support from the European Research Council (ERC Grant Agreement no 337603), the Danish Council for Independent Research (Sapere Aude), the Swiss National Science Foundation (project no PP00P2 150734), and the VILLUM FONDEN via the QMATH Centre of Excellence (Grant No. 10059). IN's research has been supported by the ANR project {StoQ} {ANR-14-CE25-0003-01}, as well as by a von Humboldt fellowship.

\section{Preliminaries}

In this paper, we shall fix a basis $\{e_1, e_2, \ldots, e_n\}$ of $\mathbb C^n$, and we put $e_0:=e_n$. Given an operator $U \in X\in \mathcal{M}_n(\mathbb C) \otimes \mathcal{M}_m(\mathbb C)$, we denote by $X_{ij}$ its $m \times m$ blocks:
$$X = \sum_{i,j=1}^n e_ie_j^* \otimes X_{ij}.$$

\begin{proposition}\label{prop:rank-blocks}
For any operator $X \in \mathcal M_n(\mathbb C) \otimes \mathcal M_n(\mathbb C)$ having blocks $X_{ij}$, we have
$$\Omega(X)  = \dim \operatorname{span}\{X_{ij}\}_{i,j=1}^n.$$
\end{proposition}
\begin{proof}
We use the \emph{vectorization} operation, which ``flattens'' matrices to vectors:
$$\operatorname{vec}(e_ie_j^*) := e_i \otimes e_j.$$
Applying this operation to both factors of the tensor product, the expression
$$\tilde X = \sum_{i,j=1}^n (e_i \otimes e_j) \cdot\operatorname{vec}(X_{ij})^*$$
defines an operator $\tilde X$ with the property that $\operatorname{rk}(\tilde X) = \Omega(X)$. Since the vectors $e_i \otimes e_j$ form an orthonormal basis of $\mathbb C^n \otimes \mathbb C^n$, we have that
$$\Omega(X) = \dim \operatorname{span}\{\operatorname{vec}(X_{ij})\}_{i,j=1}^n = \dim \operatorname{span}\{X_{ij}\}_{i,j=1}^n.$$
\end{proof}
\begin{remark}
One can easily show that the operator Schmidt rank of an operator $X$ is also equal to the (usual) rank of the \emph{realignment} (or reshuffling) $X^R$ of $X$ (see e.g.~\cite[Section 10.2]{bzy} for the definitions). 
\end{remark}

\begin{proposition}\label{prop:diag-blocks}
For any rank $1 \leq r \leq n$, there exists a unitary operator $U \in \mathcal U(\mathbb C^n \otimes \mathbb C^n)$ such that $\operatorname{rk} U^R = r$. 
\end{proposition}
\begin{proof}
Choose $r$ linearly independent unitary operators $V_1, \ldots, V_r \in \mathcal U_n$ (say, any $r$ distinct Weyl operators \cite[Section 4]{fho}) and define
$$U = \sum_{i=1}^r e_ie_i^* \otimes V_i + \sum_{i=r+1}^n e_ie_i^* \otimes V_r.$$
Using Proposition \ref{prop:rank-blocks}, we have
$$\Omega(U) = \dim \operatorname{span}\{V_i\}_{i=1}^r = r.$$
\end{proof}

\section{Unitaries from permutations}

In this section we study the operator Schmidt ranks of unitary operators associated to permutations. The construction below contains all but two of the possible ranks from Theorem \ref{thm:main}. We think that the construction is interesting on its own, and raises non-trivial combinatorial questions regarding pairs of permutation tableaux; see e.g.~\cite{czu} for a connection with $r$-orthogonal Latin squares. We shall denote by $\mathcal S_n$ the symmetric group on $n$ elements. 

\begin{definition}
To any $\alpha,\beta\in\mathcal{S}^n_n$ we associate the unitary operator $U_{\alpha,\beta}\in\mathcal{U}\lb\C^n\otimes\C^n\rb$ by 
\[
U_{\alpha,\beta} = \sum_{i,j=0}^{n-1} e_i e_j^* \otimes e_{\alpha_i(j)} e_{\beta_j(i)}^*.
\]
\label{defn:PermUnitary}
\end{definition}

It is easy to check that the operator $U_{\alpha,\beta}$ is indeed unitary (see \cite[Lemma 3.3]{bne} for a proof). One can compute the operator Schmidt rank of $U_{\alpha,\beta}$ from the permutations alone. For $\alpha,\beta\in\mathcal{S}^n_n$ we define 
\[
N(\alpha,\beta) := \Big{|}\left\lbrace\left(\alpha_i(j),\beta_j(i)\right) \, : \, 0 \leq i,j \leq n-1\right\rbrace\Big{|}.
\]
Then, using the fact that the blocks of the operator $U_{\alpha,\beta}$ have just one non-zero entry, it is immediate to show the following lemma:

\begin{lemma}
For any $\alpha,\beta\in\mathcal{S}^n_n$ we have $\Omega\lb U_{\alpha,\beta}\rb = N(\alpha,\beta).$
\label{lem:CountingPerm}
\end{lemma}

Using the previous lemma we will prove the following theorem.

\begin{theorem}\label{thm:permutation}
For $n>2$ and any $r\in\lset n,\ldots ,n^2\rset\setminus\lset n+1,n^2-1\rset$ there exist $\alpha,\beta\in\mathcal{S}^n_n$ such that $\Omega\lb U_{\alpha,\beta}\rb = r$.
\end{theorem}

The proof of Theorem \ref{thm:permutation} involves several constructions of permutations $\alpha,\beta\in\mathcal{S}^n_{n}$ presented in the following lemmata. 
We shall use $c\in\mathcal{S}_n$ to denote the cyclic shift $c(i) = i+1\text{ mod }n$. 

\begin{lemma}
For $n\in\N$ and any $r\in\lset n,\ldots ,2n\rset\setminus\lset n+1\rset$ there exist $\alpha,\beta\in\mathcal{S}^n_n$ such that $N(\alpha,\beta) = r$.
\label{lem:Construction1}
\end{lemma} 
\begin{proof}
For some permutation $\pi\in\mathcal{S}_n$ consider $\alpha,\beta\in\mathcal{S}^n_n$ given by
\begin{align*}
&\alpha_0 = \pi,~ \alpha_1 = c^1,~ \ldots ,~ \alpha_{n-1} = c^{n-1}\\
&\beta_0 = \text{id},~ \beta_1 = c^1,~ \ldots,~ \beta_{n-1} =c^{n-1}.
\end{align*}
With these permutations we have 
\[
\lb\alpha_i(j),\beta_j(i)\rb = \begin{cases} \lb \pi(j),j\rb, &\text{ if }i=0 \\
\lb i+j\text{ mod } n, i+j\text{ mod }n\rb, &\text{ else. } \end{cases}
\]
Let $l\in\lset 0,\ldots ,n\rset\setminus\lset n-1\rset$ be the number of fixed points of $\pi$, then we have 
\[
N(\alpha,\beta) = \Big{|}\left\lbrace(\alpha_i(j),\beta_j(i)) \, : \, 0 \leq i,j \leq n-1\right\rbrace\Big{|} = 2n -l. 
\]
As there are permutations $\pi\in\mathcal{S}_n$ having $l$ fixed points for any $l \in \lset 1,\ldots ,n\rset\setminus\lset n-1\rset$, the proof is finished.
\end{proof} 

The next lemma will be proved using a similar, but slightly more complicated construction.

\begin{lemma}
For $n>3$ and any $r\in\lset 2n,\ldots ,3n-1\rset$ there exist $\alpha,\beta\in\mathcal{S}^n_n$ such that $N(\alpha,\beta) = r$.
\label{lem:Construction2}
\end{lemma}
\begin{proof}
For some permutation $\pi\in\mathcal{S}_n$ consider $\alpha,\beta\in\mathcal{S}^n_n$ given by\footnote{Here $\alpha_0 = \alpha_1$ is \emph{not} a typo.}
\begin{align*}
&\alpha_0 = c^1, ~ \alpha_1 = c^1,~ \alpha_2 = c^2,~ \ldots ,~ \alpha_{n-1} = c^{n-1}\\
&\beta_0 = \pi,~ \beta_1 = c^1,~ \beta_2 = c^2,~ \ldots,~ \beta_{n-1} =c^{n-1}.
\end{align*}
With these permutations we have 
\begin{equation}
\lb\alpha_i(j),\beta_j(i)\rb = \begin{cases} \lb 1,\pi(0)\rb, &\text{ if }i=0, j=0 \\
\lb i,\pi(i)\rb, &\text{ if }i\neq 0, j=0\\
\lb j+1\text{ mod } n, j\rb, &\text{ if }i=0, j\neq 0 \\
\lb i+j\text{ mod } n, i+j\text{ mod }n\rb, &\text{ else. } \end{cases}
\label{equ:manyCases}
\end{equation}
Note that the last two cases in \eqref{equ:manyCases} include $2n-1$ different elements of the set 
\[\left\lbrace(\alpha_i(j),\beta_j(i)) \, : \, 0 \leq i,j \leq n-1\right\rbrace .\] 
Therefore, choosing $\pi = c^2$ leads to permutations $\alpha,\beta\in\mathcal{S}^n_n$ with $N(\alpha,\beta)=3n -1$ as the first two cases in \eqref{equ:manyCases} give different elements than the last two cases. Similarly, setting $\pi\in\mathcal{S}_n$ to
\[
\pi(i) =\begin{cases} 1, &\text{ if } i=0\\
0, &\text{ if } i=1\\
i, &\text{ else,}\end{cases} 
\]
in \eqref{equ:manyCases} gives $N(\alpha,\beta) = 2n$. Next, choosing $\pi\in\mathcal{S}_n$ such that
\[
\pi(i) =\begin{cases} 1, &\text{ if } i=0\\
2, &\text{ if } i=1\\
0, &\text{ if } i=2\\
i, &\text{ else,}\end{cases} 
\]
gives $N(\alpha,\beta) = 2n+1$. Finally, for any $0<k < n-2$ consider the permutation $\pi\in\mathcal{S}_n$ given by
\[
\pi(i) =\begin{cases} 1, &\text{ if } i=0\\
k+1, &\text{ if } i=1\\
0, &\text{ if } i=n-1\\
i+1, &\text{ if } i\in\lset k+1,\ldots ,n-2\rset\\
i, &\text{ if } i\in\lset 2,\ldots , k\rset.\end{cases} 
\] 
By inserting this into \eqref{equ:manyCases} it can be verified easily that $N(\alpha,\beta) = 3n-k-1$. 
\end{proof}


With the previous lemmata we can proof Theorem \ref{thm:permutation}.

\begin{proof}[Proof of Theorem \ref{thm:permutation}]
The proof will be done by induction in $n$. Consider the case $n=3$: By Lemma \ref{lem:Construction1} we obtain permutations $\alpha,\beta\in\mathcal{S}^3_3$ for any $r\in\lset 3,5,6\rset$ such that $N(\alpha,\beta) = r$. It is easy to verify that the choice $\alpha_i = \beta_i = \text{id}$ for any $i\in\lset 0,1,2\rset$ gives $N(\alpha,\beta) = 9$. Finally choosing $\alpha_i=\beta_i = \text{id}$ for $i\in\lset 0,1\rset$ and $\alpha_2 = \beta_2 = c^1$ gives $N(\alpha,\beta) = 7$.

Now for $n>3$ assume that for any $r\in\lset n-1,\ldots ,(n-1)^2\rset\setminus \lset n , (n-1)^2-1\rset$ there are $\alpha,\beta\in\mathcal{S}^{n-1}_{n-1}$ such that $N(\alpha,\beta) = r$. Now consider any $r'\in\lset n,\ldots , n^2\rset\setminus \lset n+1,n^2-1\rset$. In the cases where $r'\in\lset n,\ldots , 2n\rset\setminus\lset n+1\rset$ or  $r'\in\lset 2n,\ldots 3n-1\rset$ we can use Lemma \ref{lem:Construction1} and Lemma \ref{lem:Construction2} respectively to construct permutations $\alpha,\beta\in\mathcal{S}^n_n$ with $N(\alpha,\beta) = r'$. For $r'\in\lset 3n,\ldots ,n^2\rset\setminus\lset n^2-1\rset$ we set $r = r'-2n+1$ and note that $r\in\lset  n+1,\ldots ,(n-1)^2 \rset\setminus\lset (n-1)^2-1\rset$. By induction hypothesis there are $\alpha,\beta\in\mathcal{S}^{n-1}_{n-1}$ such that $N(\alpha,\beta) = r$. For any additional permutations $\alpha_{n},\beta_n\in\mathcal{S}_{n-1}$ we define $\alpha',\beta'\in\mathcal{S}^n_{n}$ by 
\[
\alpha_k'(i) = \begin{cases}\alpha_k(i),&\text{ if } i\leq n-1 \\
n,&\text{ if } i=n\end{cases},\hspace*{1cm} \beta_k'(i) = \begin{cases}\beta_k(i),&\text{ if } i\leq n-1 \\
n,&\text{ if } i=n\end{cases}
\]
for $k\leq n$. Now it is easy to verify that 
\[
N(\alpha',\beta') = \Big{|}\left\lbrace(\alpha'_i(j),\beta'_j(i)) \, : \, 0 \leq i,j \leq n\right\rbrace\Big{|} = r + 2n - 1 = r'.
\]
\end{proof}

It is remarkable that the permutation construction detailed in this section shows that all, except two ranks are possible. Note that in the case $n=2$, the remaining two cases, $n+1$ and $n^2-1$, correspond to the restriction $\Omega(U) \neq 3$ from \cite{dvc}. In the next two sections, we construct unitary operators having operator Schmidt rank equal to $n+1$ and $n^2-1$ respectively for any $n\geq 3$.

\section{The product construction}

Here we generalize a construction by Tyson from \cite{tys}. For $n\in\N$ and $k\in \mathbb Z$ define the $k$-cyclic shift $S^n_k :\C^n\ra\C^n$ on the computational basis by
\[
S^n_k e_i  = e_{i+k~\text{mod } n}
\] 
for $i\in\lset 1,\ldots n\rset$.

\begin{theorem}\label{thm:product}
For any $k,l\in\lset 1,\ldots ,\min\lb n,m\rb\rset$ there is a unitary $U\in\mathcal{U}\lb\C^n\otimes \C^m\rb$ such that $\Omega\lb U\rb = kl$.
\end{theorem}
\begin{proof}
Consider two partitions of the identity operator into orthogonal, diagonal (in the computational basis) projectors $\lset P_i\rset^k_{i=1}$ and $\lset Q_j\rset^l_{j=1}$ on $\C^n$ and $\C^m$ respectively (i.e.~$\sum^k_{i=1} P_i = \one_n$ and $\sum^l_{j=1} Q_j = \one_m$). Consider the matrices 
\[
V_1 = \sum^k_{i=1} P_i\otimes S^{m}_i~\text{ and }~V_2 = \sum^l_{j=1} S^{n}_j\otimes Q_j.
\]
Using orthogonality of the projectors and unitarity of the cyclic shifts it is easy to verify that $V_1,V_2\in\mathcal{U}\lb\C^{n}\otimes \C^{m}\rb$. This implies that 
\[
U = V_1V_2 = \sum^k_{i=1}\sum^l_{j=1} P_i S^{n}_j\otimes S^{m}_i Q_j
\] 
is a unitary matrix. Note that $\lset P_i S^{n}_j\rset_{i,j}$ and $\lset S^{m}_i Q_j\rset_{i,j}$ are orthogonal sets with respect to the Hilbert-Schmidt inner product: 
$$\langle P_i S^{n}_j,  P_{i'} S^{n}_{j'} \rangle = \operatorname{tr}\left[ S^{n}_{-j}P_{i} P_{i'} S^{n}_{j'}\right] = \delta_{i,i'}\operatorname{tr}\left[ P_{i}  S^{n}_{j'-j}\right] = \delta_{i,i'}\delta_{j,j'},$$
where for the last equality we have used the fact that the diagonal of the operator $S^{n}_{j'-j}$ is zero, unless $j=j'~\text{mod } n$. This shows that $\Omega\lb U\rb = kl$.
\end{proof}

\begin{corollary}\label{cor:product-odd}
For any odd $n\geq 3$, there exist a unitary operator $U \in \mathcal U(\mathbb C^n \otimes \mathbb C^n)$ with operator Schmidt rank $\Omega(U) = n+1$. 
\end{corollary}

The next result contains the other half of the existence proof for the operator Schmidt rank $n+1$. 

\begin{proposition}\label{prop:product-even}
For any even $n > 2$ there exists a unitary $U\in\mathcal{U}\lb\C^n\otimes \C^n\rb$ with operator Schmidt rank $\Omega\lb U\rb = n+1$.
\end{proposition}
\begin{proof}
Let $n=2k$ for $k\in\N$. By Theorem \ref{thm:product} there exists a unitary $\tilde{U}\in\mathcal{U}\lb\C^k\otimes\C^k\rb$ with operator Schmidt rank $\Omega\lb \tilde{U}\rb = n = 2k \leq k^2$. We write 
\[
\tilde{U} = \begin{pmatrix}
  A_{11} & \cdots & A_{1k} \\
  \vdots  & \ddots & \vdots  \\
  A_{k1} & \cdots & A_{kk} 
 \end{pmatrix}
\]
and note that by Proposition \ref{prop:diag-blocks} we have $n = \Omega(\tilde{U}) = \text{dim}\lb\text{span}\lset A_{ij}\rset^k_{i,j=1}\rb$. Now consider the unitary $U\in\mathcal{U}\lb\C^n\otimes \C^n\rb$ given by
\[
U = \begin{pmatrix}
  \begin{array}{c c}A_{11} & 0 \\ 0 & A_{11}\end{array} & \cdots & \begin{array}{c c}A_{1k} & 0 \\ 0 & A_{1k}\end{array} & \begin{array}{c c} 0 & 0 \\ 0 & 0\end{array} & \cdots &\begin{array}{c c} 0 & 0 \\ 0 & 0\end{array} \\
  \vdots & \ddots & \vdots & \vdots & &  \vdots \\
   \begin{array}{c c}A_{k1} & 0 \\ 0 & A_{k1}\end{array} & \cdots & \begin{array}{c c}A_{kk} & 0 \\ 0 & A_{kk}\end{array} & \begin{array}{c c} 0 & 0 \\ 0 & 0\end{array} & \cdots &\begin{array}{c c} 0 & 0 \\ 0 & 0\end{array} \\
   \begin{array}{c c}0 & 0 \\ 0 & 0\end{array} & \cdots & \begin{array}{c c}0 & 0 \\ 0 & 0\end{array} & \begin{array}{c c} V & 0 \\ 0 & W\end{array} & \cdots &\begin{array}{c c} 0 & 0 \\ 0 & 0\end{array} \\
   \vdots &  & \vdots & \vdots & \ddots &\vdots\\
   \begin{array}{c c}0 & 0 \\ 0 & 0\end{array} & \cdots & \begin{array}{c c}0 & 0 \\ 0 & 0\end{array} & \begin{array}{c c} 0 & 0 \\ 0 & 0\end{array} & \cdots &\begin{array}{c c} V & 0 \\ 0 & W\end{array}
 \end{pmatrix}
\]
for unitaries $V\neq W$, each appearing $k$ times on the diagonal of $U$. Again by Proposition \ref{prop:rank-blocks} and the choice of the $A_{ij}$ we have 
\[\Omega(U) = \text{dim}\lb\text{span}\lb\left\lbrace\begin{pmatrix} A_{ij} & 0 \\ 0 & A_{ij}\end{pmatrix} \right\rbrace^k_{i,j=1}\cup \left\lbrace \begin{pmatrix} V & 0 \\ 0 & W\end{pmatrix} \right\rbrace\rb\rb = n+1 . \] 
\end{proof}

\section{The Fourier construction}

To construct unitaries $U\in\mathcal{U}\lb \C^{n}\otimes \C^n\rb$ with $\Omega\lb U\rb = n^2 -1$ we will apply the Fourier method developed by Tyson in \cite{tys}. For $\lambda\in \mathcal M_n(\mathbb C)$ let
\[
\hat{\lambda}(a,b) = \frac{1}{n}\sum^{n-1}_{\alpha = 0}\sum^{n-1}_{\beta = 0} \exp\lb \frac{2\pi i}{n}(a\alpha + b\beta)\rb\lambda(\alpha,\beta).
\] 
denote the discrete Fourier transform. For the construction we define an orthonormal basis of $\C^n\otimes \C^n$ by 
\[
v_{\alpha,\beta} := n^{-1/2}\sum^{n-1}_{j=0} \phi^{\alpha(j-\beta)}_{n} e_{j-\beta \text{ mod }n}\otimes e_j
\]
for $\phi_n = \exp(\frac{2\pi i}{n})$. We will apply the following result of Tyson:
\begin{theorem}[{\cite[Theorem 7]{tys}}]
For $\lambda\in\mathcal M_n(\mathbb C)$ with $|\lambda(i,j)| = 1$ the unitary operator 
\[
D_\lambda = \sum^{n-1}_{\alpha=0}\sum^{n-1}_{\beta=0}\lambda\lb\alpha,\beta\rb v_{\alpha,\beta}v_{\alpha,\beta}^*~\in~\mathcal{U}\lb\C^n\otimes \C^n\rb
\]
has operator Schmidt rank
\[
\Omega\lb D_\lambda\rb = \left| \lset (a,b)\in\lset 0,\ldots ,n-1\rset^2 \, : \, \hat{\lambda}(a,b)\neq 0 \rset  \right|.
\]
\label{thm:Fourier}
\end{theorem}

With the above we have the following result, showing the existence of a unitary with operator Schmidt rank equal to $n^2-1$:
\begin{proposition}\label{prop:existence-fourier}
For any $n>2$ there exist $\lambda\in\C^{n\times n}$ with $|\lambda_{ij}| = 1$ such that the unitary operator $D_\lambda\in\mathcal{U}\lb\C^n\otimes \C^n\rb$ defined in Theorem \ref{thm:Fourier} fulfills $\Omega\lb D_\lambda\rb = n^2 -1$. 
\end{proposition}
\begin{proof}
Let $p\in \lset 1,\ldots ,n-1\rset$ be a prime number\footnote{If there were no such prime for some $n>3$, then the primes $2$ and $3$ have to divide $n$. Now consider a prime $p\in\lbr \frac{n}{2},n\rbr$ (such a prime exists by Bertrand's postulate). By assumption $p$ divides $n$  as well and we have $1\leq \frac{n}{p}\leq 2$ by construction. But then $3$ cannot divide $\frac{n}{p}$ contradicting that $3$ divides $n$.} not dividing $n$. For $i<j$ fix distinct numbers $M_{i,j}\in\N$. Now for some $x\in\R$ to be specified later choose 
\[
\lambda_x\lb \alpha,\beta\rb = \begin{cases} \phi^\alpha_n, &\text{ if } \alpha = \beta \\
(\phi^x_n)^{M_{\alpha,\beta}}, &\text{ if } \alpha>\beta \\
-(\phi^x_n)^{M_{1,0}}, &\text{ if } (\alpha,\beta) = (1,p)\\
-(\phi^x_n)^{M_{p,1}}, &\text{ if } (\alpha,\beta) = (0,1) \\
-(\phi^x_n)^{M_{\beta,\alpha}} &\text{ else.}\end{cases}
\]
We will show that $x\in\R$ can be chosen such that
\[
\left| \lset (a,b)\in\lset 0,\ldots ,n-1\rset^2 \, : \, \hat{\lambda}(a,b)\neq 0 \rset  \right| = n^2-1.
\]
Then, by Theorem \ref{thm:Fourier} the unitary $D_{\lambda_x}$ will have operator Schmidt rank $\Omega\lb D_\lambda\rb = n^2 -1$ finishing the proof. By construction we have 
\[
\hat{\lambda}_x(0,0) = \frac{1}{n}\sum^{n-1}_{\alpha = 0}\sum^{n-1}_{\beta = 0}\lambda_x(\alpha,\beta) = \sum^{n-1}_{\alpha=0}\phi^\alpha_n = 0
\]
for any $x\in\R$. Thus, we have to find $x\in\R$ with $\hat{\lambda}_x(a,b)\neq 0$ for any $(a,b)\in\lset 0,\ldots ,n-1\rset^2\setminus \lset (0,0)\rset$. For each $(a,b)\neq (0,0)$ we define a polynomial $P_{a,b}$ by 
\[
P_{a,b}(z) = \frac{1}{n}\lbr \sum^{n-1}_{\alpha = 0} \phi^{(a+b+1)\alpha}_n + \sum_{\alpha>\beta} c_{\alpha,\beta} z^{M_{\alpha,\beta}} \rbr
\]
with coefficients
\[
c_{\alpha,\beta} = \begin{cases} \exp\lb\frac{2\pi i}{n}(a\alpha + b\beta)\rb - \exp\lb\frac{2\pi i}{n} (a\beta + b\alpha)\rb, &\text{ if } (p,1)\neq (\alpha,\beta) \neq (1,0) \\
\exp\lb\frac{2\pi i}{n}(ap + b)\rb - \exp\lb\frac{2\pi i}{n} b\rb, &\text{ if } (\alpha,\beta) = (p,1)\\
\exp\lb\frac{2\pi i}{n}a\rb - \exp\lb\frac{2\pi i}{n}(a+bp)\rb, &\text{ if } (\alpha,\beta) = (1,0)\end{cases}.
\]
Note that $c_{p,1} = c_{1,0} = 0$ iff $n$ divides both $ap$ and $bp$. This can only happen for $(a,b) = (0,0)$ as $a,b<n$ and $p$ is not a divisor of $n$. Therefore, we have $P_{a,b}\neq 0$ (as polynomials) for any $(a,b)\neq (0,0)$. Finally note that 
\[
\hat{\lambda}_x(a,b) = P_{a,b}\lb\phi^{x}_n \rb.
\]
As any $P_{a,b}$ has only finitely many zeros it is possible to choose $x\in\R$ such that $\hat{\lambda}_x(a,b)\neq 0$ for all $(a,b)\neq 0$.

\end{proof}

Note that the unitary matrices constructed in this section have complex entries, whereas our other constructions are all real. This raises the question of constructing bipartite orthogonal operators $U\in\mathcal{O}\lb \C^{n}\otimes \C^{n}\rb$ having operator Schmidt rank $n^2-1$. We leave this question open. Such a construction would prove a real (i.e.~orthogonal) version of Theorem \ref{thm:main}.

\section{Unbalanced dimensions}

In the previous sections, we have considered operators $U \in \mathcal U(\mathbb C^n \otimes \mathbb C^n)$ and we have shown (Proposition \ref{prop:diag-blocks}, Theorem \ref{thm:permutation}, Corollary \ref{cor:product-odd}, Proposition \ref{prop:product-even}, Proposition \ref{prop:existence-fourier}) that, for $n \geq 3$, any rank $r \in \{1, \ldots, n^2\}$ is achievable. To finish the proof of Theorem \ref{thm:main}, we show in this section how to extend the previous constructions to the general case $n \neq m$.

\begin{theorem}\label{thm:unbalanced}
For any $(n,m)\neq (2,2)$ and $r\in\lset 1,\ldots ,\min\lb n,m\rb^2\rset$ there exists a unitary $U\in\mathcal{U}\lb\C^{n}\otimes \C^m\rb$ with operator Schmidt rank $\Omega\lb U\rb = r$.

\end{theorem}
\begin{proof}
Without loss of generality we can assume $n<m$. As $\Omega(V\otimes W) = 1$ for $V\in\mathcal{U}\lb\C^n\rb$ and $W\in\mathcal{U}\lb\C^m\rb$ we will assume in the following that $r>1$. 
We begin with the case of $n>2$: By the previous results dealing with the $n=m$ case, there exists a unitary $\tilde{U}\in\mathcal{U}\lb \C^{n}\otimes \C^n\rb$ with operator Schmidt rank $\Omega\lb \tilde{U}\rb = r-1$. We write 
\begin{equation}
\tilde{U} = \begin{pmatrix}
  A_{11} & \cdots & A_{1n} \\
  \vdots  & \ddots & \vdots  \\
  A_{n1} & \cdots & A_{nn} 
 \end{pmatrix}
\label{equ:blockform}
\end{equation}
and note that by Proposition \ref{prop:rank-blocks} we have $r-1 = \Omega(\tilde{U}) = \text{dim}\lb\text{span}\lset A_{ij}\rset^n_{i,j=1}\rb$. As $r-1\leq n^2 -1$ there exists a unitary $V\in \mathcal{U}\lb\C^n\rb$ with $V\notin \text{span}\lset A_{ij}\rset^n_{i,j=1}$. Now consider $U\in\mathcal{U}\lb\C^n\otimes \C^m\rb$ given by
\[
U = \begin{pmatrix}
  A_{11} & \cdots & A_{1n} & 0 & \cdots & 0 \\
  \vdots  & \ddots & \vdots & \vdots & & \vdots \\
  A_{n1} & \cdots & A_{nn} & 0 & \cdots & 0 \\
  0 & \cdots & 0 & V & \cdots & 0 \\
  \vdots  & \ddots & \vdots & \vdots & \ddots & \vdots \\ 
  0 & \cdots & 0 & 0 & \cdots & V.
 \end{pmatrix}
\]  
By Proposition \ref{prop:rank-blocks} and the choice of $V$ we have $\Omega(U) = \text{dim}\lb\text{span}\lb\lset A_{ij}\rset^n_{i,j=1}\cup \lset V\rset\rb\rb = r$.

In the case $n=2$ we can use the same construction as above to construct unitaries $U\in\mathcal{U}\lb\C^2\otimes\C^m\rb$ (for $m>2$) with operator Schmidt ranks $\Omega\lb U\rb\in \lset 1,2,3\rset$. The existence of a unitary $U\in\mathcal{U}\lb\C^2\otimes\C^m\rb$ with $\Omega\lb U\rb = 4$ follows from Theorem \ref{thm:product}.


\end{proof}

\end{document}